\documentclass[12p]{article}
\usepackage{amsmath,amssymb,amsfonts,amsthm}
\usepackage{fullpage}

\usepackage[ocgcolorlinks]{hyperref}
\usepackage{xspace}

\title{On subexponential running times for approximating directed Steiner tree and related problems}

\author{
Marek Cygan\thanks{
   Department of Math and information, University of Warsaw, Warsaw, Porland.
   Email: \texttt{cygan@nimuw.edu.pl}}
\and
Guy Kortsarz\thanks{
   Computer Science Department, Rutgers University -- Camden, Camden NJ, USA.
   Email: \texttt{guyk@camden.rutgers.edu}}
\and
Bundit Laekhanukit\thanks{
   Institute for Theoretical Computer Science,
   Shanghai University of Finance and Economics, Shanghai, China.
   Email: \texttt{bundit@sufe.edu.cn}}
}

\date{\today}
\newtheorem{lemma}{Lemma}[section]
\newtheorem{theorem}[lemma]{Theorem}

\newtheorem{definition}[lemma]{Definition}
\newtheorem{corollary}[lemma]{Corollary}

\newtheorem{claim}[lemma]{Claim}

\newtheorem{conjecture}[lemma]{Conjecture}

\newcommand{\opt}{\textsc{opt}\xspace}
\newcommand{\ethh}{\textsc{ETH}\xspace}

\newcommand{\pgc}{\textsc{PGC}\xspace}
\newcommand{\sat}{\textsc{SAT}\xspace}
\newcommand{\setcover}{\textsc{Set-Cover}\xspace}
\newcommand{\submodularcover}{\textsc{Submodular-Cover}\xspace}
\newcommand{\dst}{\textsc{Directed-Steiner-Tree}\xspace}
\newcommand{\clique}{\textsc{Maximum Clique}}

\newcommand{\NP}{\textsc{NP}}
\newcommand{\ZPTIME}{\textsc{ZPTIME}}

\newcommand{\lc}{\textsc{Label-Cover}\xspace}
\newcommand{\polylog}{\mathrm{polylog}}
\newcommand{\poly}{\mathrm{poly}}
\newcommand{\gst}{\textsc{Group-Steiner-Tree}\xspace}
\newcommand{\cst}{\textsc{Covering-Steiner-Tree}\xspace}
\newcommand{\cp}{\textsc{Connected-Polymatroid}\xspace}

\newcommand{\Oh}{{\mathcal{O}}}
\newcommand{\cG}{{\mathcal{G}}}
\newcommand{\cP}{{\mathcal{P}}}

\begin{document}

%

\maketitle

\begin{abstract}
This paper concerns proving almost tight
(super-polynomial) running times, for achieving desired 
approximation ratios for various problems.
To illustrate the question we study, 
let us consider the $\setcover$ problem with $n$ elements and $m$ sets.
Now we specify our goal to approximate $\setcover$ to a factor of $(1-\alpha)\ln n$,
for a given parameter $0<\alpha<1$.
What is the best possible running time for achieving such approximation ratio?
This question was answered implicitly in the work
of Moshkovitz [Theory of Computing, 2015]:
Assuming both the {\em Projection Games Conjecture} (PGC)
and the {\em Exponential-Time Hypothesis} (ETH),
any $((1-\alpha) \ln n)$-approximation algorithm for $\setcover$
must run in time at least $2^{n^{c\cdot\alpha}}$,
for some small constant $0<c<1$. 

We study the questions along this line. 
Our first contribution is in strengthening the above result.
We show that under ETH and PGC the running time requires for  any
$((1-\alpha) \ln n)$-approximation algorithm for $\setcover$ 
is essentially $2^{n^{\alpha}}$.
This (almost) settles the question since our lower bound
matches the best known running time of $2^{O(n^\alpha)}$
for approximating $\setcover$ to within a factor $(1-\alpha) \ln n$
given by Cygan et al. [IPL, 2009].
Our result is tight up to the constant multiplying the $n^{\alpha}$ terms
in the exponent.

The lower bound of $\setcover$ applies to all of its generalization, e.g.,
$\gst$, $\dst$, $\cst$ and $\cp$. 
We show that, surprisingly, in almost exponential running time, 
these problems reduce to $\setcover$.
Specifically, we complement our lower bound by presenting 
an $(1-\alpha)\ln n$ approximation algorithm 
for all aforementioned problems that runs in time
$2^{n^{\alpha}\cdot \log n }\cdot \poly(m)$.

We further study the approximation ratio in the regime of
$\log^{2-\delta}n$ for $\gst$ and $\cst$.
Chekuri and Pal [FOCS, 2005]
showed that $\gst$ admits $(\log^{2-\alpha}n)$-approximation in 
time $\exp(2^{\log^{\alpha+o(1)}n})$, for any parameter $0 < \alpha < 1$.
We show the running time lower bound of \gst:
any $(\log^{2-\alpha}n)$-approximation algorithm
for \gst must run in time at least $\exp((1+o(1)){\log^{\alpha-\epsilon}n})$,
for any constant $\epsilon>0$, unless the ETH is false.
Our result follows by analyzing the hardness construction of \gst
due to the work of Halperin and Krauthgamer [STOC, 2003].
The same lower and upper bounds hold for \cst.
\end{abstract}

\newpage

\section{Introduction}
\label{sec:intro}

The traditional study of {\em approximation algorithms}
concerns designing {\em algorithms that run in polynomial time}
while producing a solution whose cost is within a factor
$\alpha$ away from the optimal solution.
Once the approximation guarantees meet the barrier,
a natural question is to ask whether the approximation
ratio can be improved if the algorithms are given
running time beyond polynomial.
This has been a recent trend in designing 
approximation algorithms that allows ones to
break through the hardness barrier;
see, e.g., 
\cite{BansalCLNN17-arxiv,BonnetP18,FotakisLP16,BourgeoisEP09a,BourgeoisEP09a,BourgeoisEP11,BourgeoisEP11,CyganKW09,CyganKPW08}.

While ones ask for improving the approximation ratio,
another interesting question is to ask the converse:
Suppose the approximation ratio has been specified at the start,
what is the smallest running time required to achieve 
such approximation ratio? 
This question has recently been an active subject of study;
see, e.g., \cite{ChalermsookLN13,BonnetE0P15,BonnetLP18}.

To answer the above question, ones need complexity assumptions stronger 
than $\mathrm{P}\neq\mathrm{NP}$ as this standard assumption does not
precisely specify the running times besides polynomial versus
super-polynomial. 
The most popular and widely believed assumption is the 
{\em Exponential-Time Hypothesis} ($\ethh$), which states that 3-$\sat$
admits no $2^{o(n)}$-time algorithm.
This together with the {\em almost linear size PCP theorems}
\cite{Dinur07,MoshkovitzR10} yields many running time
lower bounds for approximation algorithms
\cite{ChalermsookLN13,BonnetE0P15,BonnetLP18}.
Let us give an example of the results of this type:

\medskip

{\bf Example:} Consider the $\clique$ problem, in which the goal
is to find a clique of maximum size in a graph $G=(V,E)$ on $n$ vertices.
This problem is known to admit no $n^{1-\epsilon}$-approximation, 
for any $\epsilon>0$, unless $\mathrm{P}=\mathrm{NP}$
\cite{Hastad96,Zuckerman07}.
Now, let us ask for an $\alpha$-approximation algorithm,
for $\alpha$ ranging from constant to $\sqrt{n}$.
There is a trivial $2^{n/\alpha}\poly(n)$-time approximation algorithm,
which is obtained by partitioning vertices of $G$ into to $\alpha$ parts
and finding a maximum clique from each part separately. 
Clearly, the maximum clique amongst these solutions is an $\alpha$-approximate solution,
and the running time is $2^{n/\alpha}\poly(n)$.
The question is whether this is the best possible running-time.
Chalermsook et al. \cite{ChalermsookLN13} showed that 
such a trivial algorithm is almost tight\footnote{
Recently, Bansal et al. \cite{BansalCLNN17-arxiv} showed that $\clique$ admits
$\alpha$-approximation in time $2^{n/\tilde{O}(\alpha\log^2\alpha)}\poly(n)$.}.
To be precise, under the $\ethh$,
there is no $\alpha$-approximation algorithm that runs in time
$2^{n^{1-\epsilon}/\alpha^{1+\epsilon}}$, for any constant $\epsilon>0$,
unless the $\ethh$ is false.

\medskip

In this paper, we consider the question along this line.
We wish to show the tight lower and upper bounds on the running times of 
polylogarithmic approximation algorithms for 
$\setcover$, $\gst$ and $\dst$ (which we will define in the next section)
and related problems.

\begin{itemize}
\item For any constant $0<\alpha<1$, what is the best possible running times 
      for $(1-\alpha)\ln{n}$-approximation algorithms for $\setcover$ and $\dst$.

\item For any constant $0<\alpha<1$, what is the best possible running time for 
      $\log^{2-\alpha}$-approximation algorithms for $\gst$.

\end{itemize}

In fact, one of our ultimate goals is to find an evidence on which 
ranges of running-times that the $\dst$ problem
admits poly-logarithmic approximations.
To be precise, we would like to partially answer the question of whether 
$\dst$ admits polylogarithmic approximations in polynomial-time,
which is a big open problem in the area.
While we are far from solving the above question,
we aim to prove possibly tight running-time for 
$\dst$ in the logarithmic range in a very fine-grained manner,
albeit assuming two strong assumptions,
the {\em Exponential-Time Hypothesis} ($\ethh$) \cite{ImpagliazzoP01,ImpagliazzoPZ01}
and the {\em Projection Game Conjectures} ($\pgc$) \cite{Moshkovitz15},
simultaneously.

\subsection{The problems studied in this paper}

\subsubsection{The $\setcover$ problem and its extensions}
In the weighted $\setcover$ problem, 
the input is a universe $U$ of size $n$ 
and a collection ${\cal S}$ of $m$ subsets of 
$U$. Each set $s\in {\cal S}$ has a cost $c(s)$.
The goal is to select a minimum cost subcollection
${\cal S}'\subseteq {\cal S}$ such that the union of the sets in ${\cal S}'$ 
spans the entire universe $U$.

The more general $\submodularcover$
problem
admits as input a universe $U$ with cost $c(x)$ on every $x\in U$.
A function is {\em submodular} if for every $S\subseteq T\subseteq V$
and for every $x\in U\setminus T$, $f(S+x)-f(S)\geq f(T+x)-f(T)$.
Let $f:2^U\mapsto R$ be a submodular non-decreasing function.
The goal in the submodular cover problem is to
minimize $c(S)$ subject to $f(S)=f(U)$. This problem strictly generalizes 
the weighted $\setcover$ problem.

The $\cp$ problem is 
the case that the elements in $U$ are leaves of 
a tree, and both the elements and tree edges have costs.
The goal is to select a set $S$ so that $f(S)=f(U)$
and that $c(S)+c(T(S))$ is minimized,
where $T(S)$ is the unique tree rooted at $r$ spanning $S$.

\subsubsection{The $\gst$ problem}
In the $\gst$ problem, the input consists an 
undirected graph with cost $c(e)$ on each edge $e\in E$,
a collection of subsets $g_1,g_2,\ldots,g_k\subseteq V$ (called {\em group})
and a special vertex $r\in V$.
The goal is is to find a minimum cost tree rooted at $r$ that 
contains at least one vertex from every group $g_i$.
In the $\cst$ problem, there is a demand 
$d_i$ for every $g_i$ and $d_i$ vertices of $g_i$ must be spanned in the tree
rooted by $r$
This $\gst$ problem strictly contains the $\setcover$ problem.
Every result for $\gst$ holds also for the $\cst$ problem
given that there is a reduction from $\cst$ to $\gst$ \cite{EvenKS02,GuptaS03}.

\subsubsection{The $\dst$ problem}
In the $\dst$ problem, the input consists of a directed graph 
with costs $c(e)$ on edges, a collection 
$S$ of terminals, and a designated root $r\in V$.
The goal is to find a minimum cost directed graph rooted at $r$
that spans $S$. 
This problem has $\gst$ as a special case.

\subsection{Related work}
\label{sec:related-work}

The $\setcover$ problem is a well-studied problem.
The first logarithmic approximation, to the best of our knowledge,
is traced back to the early work of Johnson \cite{Johnson74a}.
Many different approaches have been proposed to approximate $\setcover$, e.g.,
the dual-fitting algorithm by Chv\'atal \cite{Chvatal79};
however, all algorithms yield roughly the same approximation ratio.
The more general problem, namely, the $\submodularcover$ problem
was also shown to admit $O(\log n)$-approximation
in the work of Wolsey \cite{Wolsey82}.
The question of why all these algorithms yield the same approximation ratio
was answered by Lund and Yannakakis \cite{LundY94} who showed that the approximation ratio
$\Theta(\log n)$ is essentially the best possible unless $\mathrm{NP}\subseteq\mathrm{DTIME}(n^{\log\log{n}})$.
Subsequently, Feige \cite{Feige98} showed the more precise lower bound 
that $\setcover$ admits no $(1-\epsilon)\ln{n}$-approximation, for any $\epsilon>0$, 
unless $\mathrm{NP}\subseteq\mathrm{DTIME}(n^{\polylog(n)})$;
this assumption has been weaken to $\mathrm{P}\neq\mathrm{NP}$ by
the recent work of Dinur and Steurer \cite{DinurS14}.
These lower bounds are, however, restricted to polynomial-time algorithms.
In the regime of subexponential-time,
Cygan, Kowalik and Wykurz \cite{CyganKW09} showed that $\setcover$ admits
an approximation ratio of $(1-\alpha)\ln n$ 
in $2^{O(n^\alpha+\polylog(n))}$ time.
On the negative side, Moshkovitz~\cite{Moshkovitz15} introduced 
the {\em Projection Games Conjecture} (PGC) to prove the approximation hardness of 
$\setcover$.
Originally, the conjecture was introduced in an attempt to show 
the $(1-\epsilon)\ln n$-hardness of $\setcover$ under $\mathrm{P}\neq\mathrm{NP}$
(which is now proved by Dinur and Steurer \cite{DinurS14}).
It turns out that this implicitly implies that
$\setcover$ admits no $(1-\alpha)\ln$-approximation algorithm in
$2^{n^{O(\alpha)}}$ time under PGC and ETH.

The generalization of the $\setcover$ problem is the $\gst$ problem.
Garg, Konjevod and Ravi \cite{GargKR98} presented a novel LP rounding
algorithm to approximate $\gst$ on trees to within a factor of $O(\log^2n)$.
Using the {\em probabilistic metric-tree embedding} \cite{Bartal96,FakcharoenpholRT04},
this implies an $O(\log^3n)$-approximation algorithm for $\gst$ in general graphs.
On the negative side, Halperin and Krauthgamer showed
the lower bound of $\log^{2-\epsilon} n$ for any $\epsilon>0$ for
approximating $\gst$ on trees under the assumption that $\NP\not\subseteq \ZPTIME(n^{\polylog(n)})$.
This (almost) matches the upper bound given by the algorithm by Garg et al.
For the related problem, the {\em Connected  Polymatroid} problem was given a 
polylogarithmic approximation algorithm by C\'alinescu and Zelikovsky \cite{CalinescuZ05};
their algorithm is based on the work of Chekuri, Even and Kortsarz \cite{ChekuriEK06},
which gave a {\em combinatorial} $\polylog(n)$ approximation for $\gst$ on trees.

The problem that generalizes all the above problems is the $\dst$ problem.
The best known approximation ratio for this problem is $n^\epsilon$ for any constant $\epsilon>0$
\cite{CharikarCCDGGL99,KortsarzP99} in polynomial-time. 
In quasi-polynomial-time, $\dst$ admits an $O(\log^3 n)$-approximation algorithm.
The question of whether $\dst$ admits a polylogarithmic approximation
in polynomial-time has been a long standing open problem.

\section{Our results}
\label{sec:results}

We show that 
under the combination of $\ethh$ and $\pgc$, the running time 
for approximating $\setcover$ to within a factor of $(1-\alpha)\ln n$ 
must be at least $2^{n^\alpha}$, where $0<\alpha<1$ is a given parameter.
This improves the work of Moshkovitz who (implicitly) showed 
the running time lower bound of $2^{n^{O(\alpha)}}$.
We complement this by showing that $\dst$ admits a $(1-\alpha)\ln n$ approximation algorithm
that runs in time $2^{n^{\alpha}\cdot \log n}$ time.
Since $\dst$ is the generalization of $\setcover$, $\gst$ and $\dst$,
the lower bounds apply to all the aforementioned problems.
Hence, up to a small factor of $\log n$ in the exponent, we get tight 
running time lower bounds for approximating all these problems to within $(1-\alpha)\ln n$.
Essentially, the same algorithm and proof give the same result for the $\cp$ problem.

We also investigate the work of Chekuri and Pal \cite{ChekuriP05} who showed that,
for any constant $0<\delta<1$, $\gst$ admits a $\log^{2-\delta}n$ approximation algorithm
that runs in time $\exp(2^{(1+o(1))\log^{\delta}n})$.
We show that, for any constant $\epsilon>1$,
there is no $\log^{2-\delta-\epsilon}n$ approximation algorithm
for $\gst$ (and thus $\cst$) that runs in time $\exp(2^{(1+o(1))\log^{\delta-\epsilon}n})$.
This lower bound is nearly tight.
We note that a reduction from $\cst$ to $\gst$ was given in \cite{EvenKS02}. 
Thus, any approximation algorithm for $\gst$ also applies for $\cst$. 

\section{Formal definition of our two complexity assumptions}

\begin{definition}
In the $\lc$ problem with the projection property (a.k.a., the {\em Projection game}),
we are given 
a bipartite graph $G(A,B,E)$, two alphabet sets (also called {\em labels}) $\Sigma_A$ 
and $\Sigma_B$, and for any edge (also called {\em query})
$e\in E$, there is a function $\phi_e: \Sigma_A\mapsto \Sigma_B$.
A labeling $(\sigma_A,\sigma_B)$ is a pair of functions 
$\sigma_A:A\mapsto \Sigma_A$ and $\sigma_B:B\mapsto \Sigma_B$ 
assigning labels to each vertices of $A$ and $B$, respectively.
An edge $e=(a,b)$ is {\em covered} by $(\sigma_A,\sigma_B)$ if
$\phi_e(\sigma_A(a))=\sigma_B(b)$.
The goal in $\lc$ is to find a labeling $(\sigma_A,\sigma_B)$
that covers as many edges as possible.
\end{definition}
In the context of the {\em Two-Provers One-Round game} (2P1R), 
every label is an answer to some "question" $a$ sent to the
Player $A$ and some question $b$ sent to the Player $B$,
for a query $(a,b)\in E$.
The two answers make the verifier accept
if a label $x\in \Sigma_x$ assigned to $a$ and 
a label $y\in \Sigma_B$ assigned to $b$
satisfy $\phi(x)=y$. 
Since any label $x \in \Sigma_A$ has a unique label in $\Sigma_B$ 
that causes the verifier to accept, $y$ is called
the {\em projection of $x$ into $b$}.

We use two conjectures in our paper.
The first is the Exponential Time Hypothesis
($\ethh$),
which asserts that an instance of the 3-$\sat$ problem on $n$ variables and $m$ clauses cannot be solved in $2^{o(n)}$-time. This was later showed by Impagliazzo, Paturi and Zane \cite{ImpagliazzoP01} that any 3-$\sat$ instance can be sparsified in $2^{o(n)}$-time to an instance with $m=O(n)$ clauses. 
Thus, ETH together with the sparsification lemma \cite{ImpagliazzoPZ01}
implies the following:

\medskip
\noindent
{\bf Exponential-Time Hypothesis combined with the Sparsification Lemma:}
Given a boolean 3-CNF formula $\phi$ on $n$ variables and $m$ clauses,
there is no $2^{o(n+m)}$-time algorithm that decides whether $\phi$ is satisfiable.
In particular, 3-\sat admits no subexponential-time algorithm.
\medskip

The following was proven by Moshkovitz and Raz \cite{MoshkovitzR10}.
\begin{theorem}[\cite{MoshkovitzR10}]
\label{thm:near-linear-PCP}
There exists $c>0$, such that for every $\epsilon \ge 1/n^c$, $3$-$\sat$ on inputs of size $n$ can be efficiently reduced to $\lc$  of size $N=n^{1+o(1)}\poly(1/\epsilon)$ 
over an alphabet of size $\exp(1/\epsilon)$ that has soundness error $\epsilon$. The graph is bi-regular (namely, every two questions on the same side 
participate in the same number of queries).
\end{theorem}

There does not seem to be an {\em inherent} reason that the alphabet would be so large.
This leads to the following conjecture  posed by Moshkovitz \cite{Moshkovitz15}.

\begin{conjecture}[The Projection Games Conjecture \cite{Moshkovitz15}]
\label{conj:proj-game}
There exists $c>0$, such that for every $\epsilon \ge 1/n^c$, $3$-$\sat$ on inputs of size $n$ can be efficiently reduced to $\lc$  of size $N=n^{1+o(1)}\poly(1/\epsilon)$ 
over an alphabet of size $\poly(1/\epsilon)$ that has soundness error $\epsilon$. Moreover, the graph is bi-regular (namely, every two questions on the same side 
participate in the same number of queries).
\end{conjecture}

The difference between Theorem~\ref{thm:near-linear-PCP} and Conjecture~\ref{conj:proj-game}
is in the size of the alphabet. 

For our purposes, we only need soundness $\epsilon=1/\polylog(n)$,
and we know that the degree and alphabet size of the graph in
Conjecture~\ref{conj:proj-game} are always $\polylog(n)$
(which are inverse of the soundness).
Hence, we may assume the (slightly) weaker assumption (obtained by setting $\epsilon=1/\polylog(n)$ in Conjecture~\ref{conj:proj-game}) as below.

\begin{conjecture}[Projection Games Conjecture, a variant]
\label{conj:pgc2}
There exists $c>0$, such that for every $\epsilon =1/\polylog(n)$, $3$-$\sat$ on inputs of size $n$ can be efficiently reduced to $\lc$ of size $N=n^{1+o(1)}\poly(1/\epsilon)$ 
where {\bf the graph is bi-regular and all degrees are bounded by} 
$\polylog(n)$. 
The size of the alphabet is $\polylog(n)$ and the soundness is $1/\polylog(n)$.
and the completeness is $1$.
\end{conjecture}

We need to inspect very carefully and slightly change the proof of \cite{Moshkovitz15}
since we do not want the $\lc$ instance to grow by a lot 
by the modification in \cite{Moshkovitz15}. Hence, in fact we have to 
go over all steps of \cite{Moshkovitz15} and bound the size more carefully in all
steps that require that.

\section{First part of the proof}
We start with the same definition as in \cite{Moshkovitz15}.

\begin{definition}[Total disagreement]
Let $(G = (A,B,E), \Sigma_A, \Sigma_B, \Phi)$ be a $\lc$ instance.
Let $\phi_A:A \to \Sigma_A$ be an assignment to the A-vertices. 
We say that the A-vertices {\em totally disagree} on a vertex $b \in B$,
if there are no two neighbors $a_1, a_2 \in A$ of $b$, for which
$$\pi_{e_1}(\phi_A(a_1)) = \pi_{e_2}(\phi_A(a_2))\,,$$
where $e_1 = (a_1,b), e_2 = (a_2,b) \in E$.
\end{definition}

The above simply states that for a given assignment $\phi_A$ and a vertex $b \in B$,
no matter which label we assign to the vertex $b$, we will satisfy only one edge incident to it.

\begin{definition}[Agreement soundness]
Let $\cG = (G = (A,B,E), \Sigma_A, \Sigma_B, \Phi)$ be a $\lc$ for
deciding whether a Boolean formula $\phi$ is satisfiable.
We say that $\cG$ has {\em agreement soundness error} $\epsilon$, if
for unsatisfiable $\phi$, for any assignment $\phi_A : A \to \Sigma_A$, the $A$-vertices are in total disagreement on at least
$1-\epsilon$ fraction of the $b \in B$.
\end{definition}
For a Yes-Instance (of 3-$\sat$), a standard argument
implies that you can label the vertices so that every edge is covered.
The usual condition of soundness required is that the number of edges 
covered is a small fraction of the edges, for every label assignment.
The total disagreement is stronger than that.
It states that for any assignment $\phi_A$, no matter how we set $\phi_B$ almost all of vertices of $B$ will have at most one incident edge satisfied.

In the rest of this subsection the goal is to show (list) agreement soundness error of bounded degree $\lc$ instances.
First, we use the following lemma (we do not alter its proof).

\begin{lemma}[Combinatorial construction]
\label{lem:combinatorial}
For $0 < \epsilon < 1$, for a prime power $D$, and $\Delta$ that is a power of $D$,
there is an explicit construction of a regular bipartite graph $H = (U,V,E)$ with $|U| = n$, $V$-degree $D$, and
$V \le n^{O(1)}$ that satisfies the following. 
For every partition $U_1, \ldots, U_\ell$ of $U$ into sets such that $|U_i| \le \epsilon |U|$, 
for $i = 1, \ldots, \ell$, the fraction of vertices $v \in V$ with more than one neighbor in any single set $U_i$, is at most $\epsilon D^2$.
\end{lemma}

It is rather trivial to show the above lemma by a probabilistic method.
Moshkovitz showed in \cite{Moshkovitz15} that such graphs 
can be constructed deterministically via a simple and elegant construction.

In the next lemma, we show how to take a $\lc$ instance with standard soundness 
and convert it to a $\lc$ instance with total disagreement soundness, by combining it with the graph from Lemma~\ref{lem:combinatorial}.
Here (as opposed to \cite{Moshkovitz15}) we have to bound the size of the created instance more carefully 
(\cite{Moshkovitz15} only states that the size is raised to a constant power).

\begin{lemma}
\label{lem:agreement}
Let $D \ge 2$ be a prime power and let $\Delta$ be a power of $D$. Let $\epsilon > 0$. From a $\lc$
instance with soundness error $\epsilon^2 D^2$ and $B$-degree $n$, we can construct a $\lc$ instance with
 agreement soundness error $2\epsilon D^2$ and $B$-degree $D$. 
The transformation preserves the alphabets, and the size of the created instance is increased by 
a factor $\poly(\Delta)$, namely by polynomial in the original $B$-degree.
\end{lemma}

\begin{proof}
Let $\cG = (G = (A,B,E), \Sigma_A, \Sigma_B, \Phi)$ be the original $\lc$
from the Projection Game Conjecture.
Let $H = (U,V,E_H)$ be the graph from Lemma~\ref{lem:combinatorial}, 
where $\Delta$, $D$ and $\epsilon$ are as given in the current lemma. 
Let us use $U$ to enumerate the neighbors of a $B$-vertex, i.e., there 
is a function $E^{\leftarrow} : B \times U \to A$ that, 
given a vertex $b \in B$ and $u \in U$,
gives us the $A$-vertex which is the $u$ neighbor of $b$.

We create a new $\lc$ $(G = (A, B \times V, E'), \Sigma_A, \Sigma_B, \Phi')$.
The intended assignment to every vertex $a \in A$ is the same as its assignment 
in the original instance. 
The intended assignment to a vertex $\langle b,v \rangle \in B \times V$ is the same as the assignment 
to $b$ in the original game.
We put an edge $e' = (a, \langle b,v \rangle)$ if there exist $u\in U$
such that $E^{\leftarrow}(b,u) = a$ and $(u,v) \in E_H$.
We define $\pi_{e'} = \pi_{(a,b)}$.

If there is an assignment to the original instance that satisfies $c$ fraction of its edges, then the
corresponding assignment to the new instance satisfies $c$ fraction of its edges (this follows from the regularity of the graph $H$).

Suppose there is an assignment for the new instance $\phi_A: A \to \Sigma_A$ in which more than $2\epsilon D^2$ fraction
of the vertices in $B \times V$ do not have total disagreement.

Let us say that $b \in B$ is good if for more than an $\epsilon D^2$ fraction of the vertices in 
$\{b\} \times V$ the $A$-vertices do not totally disagree.
Note that the fraction of good $b \in B$ is at least $\epsilon D^2$.

Focus on a good $b \in B$. Consider the partition of $U$ into $|\Sigma_B|$ sets, where the set corresponding to $\sigma \in \Sigma_B$ is:
$$U_{\sigma} = \{u \in U | a = E^{\leftarrow}(b,u) \wedge e = (a,b) \wedge \pi_e(\phi_A(a)) = \sigma\}\,.$$

By the goodness of $b$ and the property of $H$, there must be $\sigma \in \Sigma_B$ such that $|U_\sigma| > \epsilon |U|$.
We call $\sigma$ the “champion” for $b$.

We define an assignment $\phi_B: B\to \Sigma_B$ that assigns good vertices $b$ their {\em champions}, and other vertices
$b$ arbitrary values. The fraction of edges that $\phi_A, \phi_B$ satisfy in the original instance is at least $\epsilon^2 D^2$.

The new instance is bigger by a factor $|V|$, which is $\poly(\Delta)$.
\end{proof}

Next we consider a variant of $\lc$ that is relevant for the reduction to 
$\setcover$.
In this variant, the prover is allowed to assign each vertex $\ell$ values, and an agreement 
is interpreted as agreement on one of the assignments in the list.

\begin{definition}[List total disagreement \cite{Moshkovitz15}]
Let $(G=(A,B,E), \Sigma_A, \Sigma_B, \Phi)$ be a \lc.
Let $\ell \ge 1$.
Let $\hat{\phi}_A :A \to \binom{\Sigma_A}{\ell}$ be an assignment
that assigns each $A$-vertex $\ell$ alphabet symbols. 
We say that the $A$-vertices {\em totally disagree} 
on a vertex $b \in B$ if there are no two neighbors
$a_1, a_2 \in A$ of b, for which there exist $\sigma_1 \in \hat{\phi}_A(a_1), \sigma_2 \in \hat{\phi}_A(a_2)$ such that
$$\pi_{e_1}(\sigma_1) = \pi_{e_2}(\sigma_2)\,,$$
 where $e_1 = (a_1,b), e_2 = (a_2,b) \in E$.
\end{definition}

\begin{definition}[List agreement soundness \cite{Moshkovitz15}]
Let $(G = (A,B,E), \Sigma_A, \Sigma_B, \Phi)$ be a \lc for
deciding membership whether a Boolean formula $\phi$ is satisfiable. 
We say that $G$ has {\em list-agreement soundness error} $(\ell, \epsilon)$, if for unsatisfiable $\phi$, for any assignment $\hat{\phi}_A : A \to \binom{\Sigma_A}{\ell}$,
the A-vertices are in total disagreement on at least $1-\epsilon$ fraction of the $b \in B$.
\end{definition}

If a PCP has low error $\epsilon$, then even when the prover is allowed to assign each $A$-vertex $\ell$ values, the game is still sound. 
This is argued in the next corollary. 

\begin{lemma}[Lemma 4.7 of \cite{Moshkovitz15}]
\label{lem:list}
Let $\ell \ge 1$, $0 < \epsilon' < 1$.
Any instance of \lc with agreement soundness error $\epsilon'$ has list-agreement soundness error $(\ell, \epsilon' \ell^2)$.
\end{lemma}

The following corollary summarizes this subsection.

\begin{corollary}
\label{cor:lc}
For any $\ell = \ell(n) = \polylog(n)$, for any constant prime power $D$ and constant $0 <\alpha <1$,
3-\sat on input of size $n$ can be reduced to a $\lc$ instance of size $N = n^{1+o(1)}$ with alphabet
size $\polylog(n)$, where the $B$-degree is $D$, and the list-agreement soundness error is $(\ell,\alpha)$.
\end{corollary}

\begin{proof}
Our starting point is the $\lc$ instance from~\ref{conj:pgc2} with soundness error $\epsilon$,
such that $2\sqrt{\epsilon} \cdot l^2 \le \alpha$. 
Note that the $B$-degree of the instance is $\Delta=\polylog(n)$.
The corollary then follows by invoking Lemma~\ref{lem:agreement} and Lemma~\ref{lem:list}.
\end{proof}

\subsection{From $\lc$ to $\setcover$}

\begin{lemma}[Partition System \cite{Moshkovitz15}]
\label{lem:partition}
For natural numbers $m$, $D$, and $0 < \alpha < 1$, for all $u \ge (D^{\Oh(\log D)} \log m)^{1/\alpha}$,
there is an explicit construction of a universe $U$ of size $u$ and partitions $\cP_1, \ldots, \cP_m$ of $U$
into $D$ sets that satisfy the following: there is no cover of $U$ with $\ell = D \ln |U| (1-\alpha)$ sets $S_{i_1}, \ldots, S_{i_\ell}$,
$1 \le i_1 < \ldots < i_{\ell} \le m$, such that each set $S_{i_j}$ belongs to the partition $P_{i_j}$.
\end{lemma}

We will use the contrapositive of the lemma: if U has a cover of size at most $\ell$, then this cover
must contain at least two sets from the same partition. Next follows the reduction, which is almost the same
as in \cite{Moshkovitz15}, where the only difference is the parameter setting.

We take a $\lc$ instance $\cG$ from Corollary~\ref{cor:lc} and transform it into an instance of $\setcover$.
In order to do so, we invoke Lemma~\ref{lem:partition} with $m = |\Sigma_B|$ and $D$ which is the $B$-degree of $\cG$.
The parameter $u$ will be determined later.
Let $U$ be the universe, and $\cP_{\sigma_1}, \ldots, \cP_{\sigma_m}$ be the partitions of $U$,
where the partitions are indexed by symbols of $\Sigma_B$.
The elements of the \setcover instance are $B \times U$, i.e., for each vertex $b \in B$ there is a copy of $U$.
Covering $\{b\} \times U$ corresponds to satisfying the edges that touch $b$.
There are $m$ ways to satisfy the edges that touch $b$ -- one for every possible assignment $\sigma \in \Sigma_B$ to $b$.
The different partitions covering $U$ correspond to those different assignments.

For every vertex $a \in A$ and an assignment $\sigma \in \Sigma_A$ to $a$, we have a set $S_{a,\sigma}$ in the $\setcover$ instance.
Taking $S_{a,\sigma}$ to the cover would correspond to assigning $\sigma$ to $a$. Notice that a cover might consist of
several sets of the form $S_{a, \cdot}$ for the same $a \in A$, which is the reason we consider list agreement. 
The set $S_{a,\sigma}$ is a union of subsets, one for every edge $e = (a,b)$ touching $a$. 
Suppose $e$ is the $i$-th edge coming into $b$ ($1 \le i \le D$), then the subset associated 
with $e$ is $\{b\} \times S$, where $S$ is the $i$-th subset of the partition $\cP_{\Phi_e(\sigma)}$.

If we have an assignment to the $A$-vertices such that all of the neighbors of $b$ agree on one value for $b$, 
then the $D$ subsets corresponding to those neighbors and their assignments form a partition that covers
$b$’s universe. On the other hand, if one uses only sets that correspond to totally disagreeing assignments
to the neighbors, then by the definition of the partitions, covering $U$ requires $\approx \ln |U|$ times more sets.
The formal claim proved by Moshkovitz is as follows.

\begin{claim}[Claim 4.10 of \cite{Moshkovitz15}]
\label{claim:reduction}
The following holds
\begin{itemize}
\item Completeness: If all the edges in $\cG$ can be satisfied, then the created instance admits a set cover of size $|A|$.
\item Soundness: Let  $\ell := |D| \ln |U|(1-\alpha)$ be as in Lemma~\ref{lem:partition}. If $\cG$ has agreement soundness $(\ell,\alpha)$, then
every set cover of the created instance is of size more than $|A|\ln |U|(1-2\alpha)$.
\end{itemize}
\end{claim}

The following is our main theorem, where we fine-tune the parameters to get the best possible (and thus almost tight)
running time lower bound.

\begin{theorem}
\label{thm:sc-main}
Fix a constant $\gamma > 0$ and $\epsilon > 0$.
Assuming \pgc there is an algorithm that given an instance $\phi$ of $3$-\sat of size $n$ 
one can create an instance $I$ of \setcover with universe of size $n^{1+o(1)} \cdot u$ 
such that if $\phi$ is satisfiable, then $I$ has a set cover of size $x$,
while if $\phi$ is not satisfiable, then $I$ does not admit a set cover of size at most $x \ln |u| (1-\epsilon)$.
\end{theorem}

\begin{proof}
Given a sparsified 3-CNF formula $\phi$ of size $n$ we transform it into a 
$\lc$ instance $\cG$,
by Corollary~\ref{cor:lc}, obtaining a list-agreement soundness error $(\ell,\alpha)$,
where we set $\alpha=\epsilon / 2$ and $\ell = |D| \ln |U|(1-\alpha)$.
Next, we perform the reduction from this section and by Claim~\ref{claim:reduction} we have the following:
\begin{itemize}
  \item If $\phi$ is satisfiable, then there exists a solution of size $|A|$ (where $A$ is one side of $\cG$).
  \item If $\phi$ is not satisfiable, then any set cover has size more than $|A| \ln |U|(1-2\alpha) = |A| \ln |U| (1-\epsilon)$.
\end{itemize}
\end{proof}

By setting the value of $|U|=u$ appropriately we get a tradeoff between
the approximation ratio and running time in the following lower bound obtained directly 
from Theorem~\ref{thm:sc-main}.

\begin{corollary}
Unless the \ethh fails, for any $0 < \alpha < 1$ and $\epsilon > 0$ there is no $(1-\alpha) \ln n$ approximation
for \setcover with universe of size $n$ and $m$ sets in time $2^{n^{\alpha - \epsilon}} \poly(m)$.
\end{corollary}

\begin{proof}
Set $u=|\phi|^{1/\alpha-1}$, then the created instance has at most $|\phi|^{1/\alpha+o(1)}$ elements,
which fits the desired running time in the lower bound.
It remains to analyze the approximation ratio.
Note that $|A| \le u^{\alpha / (1-\alpha)}$, hence 
$$(1-\alpha) \ln (|A| \cdot u) \le (1-\alpha) (\alpha / (1-\alpha)+1) \ln u = \ln u\,.$$
\end{proof}

\section{Approximating Directed Steiner Tree}

In this section, we present a $(1-\epsilon)\cdot \ln n $-approximation algorithm for \dst
running in time $2^{\Oh(n^{\alpha} \log n)}$.

\begin{lemma}
\label{lem:tree}
For any rooted tree $T$ with $\ell$ leaves, there exists
a set $X \subseteq V(T)$ of $\Oh(n^{\alpha})$ vertices
together with a family of edge disjoint trees $T_1, \ldots, T_q$,
such that:
\begin{itemize}
\item the trees are edge (but not vertex) disjoint
  \item each $T_i$ is a subtree of $T$,
  \item the root of each $T_i$ belongs to $X$,
  \item each leaf of $T$ is a leaf of exactly one $T_i$,
  \item each $T_i$ has more than $n^{\alpha}$ but less than $2n^{\alpha}$   leaves.
\end{itemize}
\end{lemma}

\begin{proof}
As long as the tree has more than $n^{\alpha}$ leaves
do the following: pick the lowest vertex $v$ in the tree,
the subtree rooted at which has more than $n^\alpha$
leaves. This implies that all its children 
contain strictly less than $n^\alpha$ leaves.
Accumulation subtrees gives at most $2n^\alpha$ leaves
since before the last iteration there were less that $n^\alpha$ leaves, 
and the last iteration adds a tree of at most $n^\alpha$ leaves.
Remove the collected tree, but do not remove their root
(namely this root may later participate in other trees).
Note that after the accumulated trees are removed, the tree rooted by our chosen root may still have more than $n^\alpha$ leaves.
This gives $\Theta(n^\alpha)$ edge disjoint trees 
with $\Theta(n^\alpha)$ leaves each. Thus, there is a tree 
with $\Theta(n^\alpha)$ leaves, and density (cost over the number of leaves)
no larger than the optimum density.
\end{proof}
For simplicity, we make sure that the number of leaves 
in each tree is exactly $n^\alpha$ by discarding leaves.
Since the trees are edge disjoint there must be a tree 
whose density: cost over the number of leaves is not worse 
(up to a factor of $2$) than the optimum density $\opt/\ell$.

Let $(G,K,r)$ be an instance of {\dst}.
Our algorithm enumerates guesses 
the (roughly) $n^\alpha$ leaves $L'$ in 
the tree whose density is no worse than the 
optimal density, and also guesses the subset $X_{L'} \subseteq V(G)$ 
that behaves as Steiner vertices. Assuming the graph went via 
a transitive closure, the size of $X_{L'}$ is at most the size of $L'$
of size $\Oh(n^\alpha)$. For a fixed set $X$,
the algorithm first finds an optimum directed Steiner
tree $T_0$. We note that assuming that the graph went via transitive closure,
we may assume that the number of Steiner vertices is less than the number of leaves,
and so we may guess the Steiner vertices at time $n^{\sqrt n}$ as well.
It is known that given the Steiner vertices and the leaves of the tree,
we can, in polynomial time, find the best density tree 
with these leaves and these $X$ vertices. 
The first such algorithm is due to
Dreyfus and Wagner \cite{DreyfusW71}.
The algorithm is quite non trivial and 
that uses dynamic programming.
The running time of the algorithm is $O(3^{n})$ time which 
is negligible in our context.

We iterate adding more trees in this way.
Each time we find the best 
density tree rooted at some vertex of $X$
which covers $n^\alpha$ leaves, 
and add the edges 
to $S$.
Each time it requires  $2^{\Oh(n^\alpha \log n})$ time.
Finally, when there are less than $(e^2+1)n^\alpha$
unconnected terminals left, we find an optimum directed
Steiner tree for those vertices.

\begin{lemma}
The approximation ratio of the above algorithm is at most $(1-\alpha)\ln n$.
\end{lemma}

\begin{proof}
Let $T$ be an optimum Steiner tree spanning $K$, let $\opt$
be its cost and let $X$ be the set from Lemma~\ref{lem:tree}
for the tree $T$.
Let us analyze the algorithm in the iteration when it
chooses the set $X$ properly, that is picks the same set as Lemma~\ref{lem:tree}.
Note that since vertices of $X$ belong to $T$
the cost of the first tree $T_0$ found by the algorithm is at most $\opt$.
The last property of Lemma~\ref{lem:tree} guarantees, that our algorithm
always finds a tree with at least as good density as $\opt/r$, where $r$
is the number of not yet connected terminals.
By the standard set-cover type analysis we can bound the cost 
of all the best density trees found by the algorithm by
$$\sum_{i=n}^{e^2\cdot 2^{1-\alpha}} \frac{\opt}{i} = \opt\cdot (H_n - 
H_{e^2\cdot 2^{1-\alpha}}) \tilde{=} (1-\alpha)\ln n\cdot \opt.$$
Finally, the last tree is of cost at most $\opt$, and it can be found in time $2^{\Oh( n^{\alpha}\log n)}$.
\end{proof}

Now we observe that the same theorem applies for 
the $\cp$ problem.
Since the function is both submodular and {\em increasing}
for every collection of pairwise disjoint sets 
$\{S_i\}_{i=1^k}$, $\sum_{i=1}^kf(S_i)\geq f(\bigcup_{i=1}^k S_i)$.
Thus for a given $\alpha$, 
at iteration $i$ there exists a collection 
of leaves $S_i$ so that 
$f(S_i)/c(S_i)\geq f(U)/c(U)$. We can guess $S_i$ 
in time $\exp(n^\alpha\cdot \log n)$ 
and its set of Steiner vertices $X_i$ in time $O(3^{n^{\alpha}})$.
Using the algorithm of \cite{DreyfusW71}, we can find a tree of
density at most $2\cdot \opt/n^{\alpha}$. The rest of the proof is identical.

\section{Hardness for Group Steiner Tree under the \ethh}
\label{sec:gst-subexpo-approx-hardness}

In this section, we show that the approximation hardness of
the group Steiner problem under the \ethh,
which implies that the subexponential-time algorithm for $\gst$ 
of Chekuri and Pal \cite{ChekuriP05} is nearly tight.
This hardness result is implicitly in the work Halperin and Krauthgamer \cite{HalperinK03}.
More precisely, the following is a corollary of Theorem~1.1 in
\cite{HalperinK03}.

\begin{theorem}[Corollary of Theorem~1.1 in \cite{HalperinK03}]
\label{thm:gst-subexpo-hardness}
Unless the \ethh is false, for any parameter $0<\delta<1$, there is no $\exp(2^{\log^{\delta-\varepsilon}N})$-time $\log^{2-\delta-\varepsilon} k$-approximation algorithm for $\gst$, for any $0<\varepsilon<\delta$.
\end{theorem}

\begin{proof}[Proof Sketch of Theorem~\ref{thm:gst-subexpo-hardness}]

We provide here the parameter translation of the reduction in \cite{HalperinK03},
which will prove Theorem~\ref{thm:gst-subexpo-hardness}.

\paragraph*{The Reduction of Halperin and Krauthgamer.}
We shall briefly describe the reduction of Halperin and Krauthgamer. 
The starting point of their reduction is the \lc instance
obtained from $\ell$ rounds of parallel repetition. 
In the first step, given a $d$-regular \lc instance $\mathcal{G}=(G=(A,B,E),\Sigma_A,\Sigma_B,\phi)$
completeness $1$ and soundness $\gamma$,
they apply $\ell$ rounds of parallel repetition
to get a $d^{\ell}$-regular instance of $\lc$ 
$\mathcal{G}'=(G^\ell=(A^\ell,B^\ell,E'),\Sigma_A^\ell,\Sigma_B^\ell,\phi')$.
To simplify the notation, we let $m=|A|=|B|$, $\sigma=|\Sigma_A^\ell|=|\Sigma_B^\ell|$
be the number of vertices and the alphabet size of the \lc instance $\mathcal{G}$, respectively.
Then we have that the number of vertices and the alphabet size of $\mathcal{G}'$ is
$2m^{\ell}$ and $\sigma^{\ell}$, respectively.
In the second step, they apply a recursive composition to produce an instance $I$ 
of \gst on a complete $(2\cdot m^{\ell}\cdot\sigma^{\ell})$-ary tree $\widehat{T}$ 
of height $H$ on $k=d^{\ell}m^{\ell\cdot H}$ groups.
Moreover, if $\mathcal{G}$ is a Yes-Instance (i.e., there is a labeling that covers all the constraints), 
then there is a feasible solution to \gst 
with cost $H$ with high probability,
and if $\mathcal{G}$ is a No-Instance (i.e., every labeing covers at most $\gamma$ fraction of the edges),
then there is no solution 
with cost less than $\beta H^2\log k$, for some sufficiently small constant $\beta>0$.

In short, the above reduction gives an instance of \gst on a tree
with $N=O((\sigma m)^{\ell H})$ vertices, $k=O(d^{\ell}m^{\ell H})$ groups,
and with approximation hardness gap $\Omega(H\log k)$.
Additionally, the reduction in \cite{HalperinK03} requires
$\ell > c_0(\log H + \log\log m + \log\log d)$ 
for some sufficiently large constant $c_0 > 0$.
\end{proof}

\paragraph*{Subexponential-Time Approximation-Hardness.}
Now we derive the subexponential-time approximation hardness for \gst.
We start by the nearly linear-size PCP theorem of Dinur~\cite{Dinur07},
which gives a reduction from 3-\sat of size $n$ 
(the number of variables plus the number of clauses)
to a label cover instance $\mathcal{G}=(G=(A,B,E),\Sigma_A,\Sigma_B,\phi)$ with 
completeness $1$, soundness $\gamma$ for some $0<\gamma<1$,
$|A|,|B| \leq  n\cdot\polylog(n)$, 
degree $d=O(1)$ and alphabet sets $\Sigma_A|,|\Sigma_B|=O(1)$.

For every parameter $0<\delta<1$, we choose $H=\log^{1/\delta - 1}{n}$,
which then forces us to choose $\ell=\Theta((1/\delta - 1)\log\log{n})$.
Note that we may assume that $\delta \ll n$ since it is a fixed parameter.
Plugging in these parameter settings, 
we have an instance of \gst on a tree with $N$ vertices 
and $k$ groups such that
\[
N=\left(O(1)\cdot{n}^{1+o(1)}\right)^{\Theta((1/\delta - 1)\log\log{n})\cdot \log^{1/\delta - 1}{n}}
 = \exp\left(\log^{1/\delta+o(1)}{n}\right)
\]
and 
\[
k=O(1)^{\Theta((1/\delta - 1)\log\log{n})}\left({n}^{1+o(1)}\right)^{\Theta((1/\delta - 1)\log\log{n})\cdot \log^{1/\delta}{n}}
 = \exp\left(\log^{1/\delta+o(1)}{n}\right)
\]
Observe that $H \geq \log^{1-\delta - o(1)}k$.
Thus, the hardness gap is 
$\Omega(H\log k)=\Omega(\log^{2 -\delta - o(1)} k)$.
This means that any algorithm for \gst on this family of instances
with approximation ratio $\log^{2-\delta-\varepsilon}k$, 
for any constant $\varepsilon>0$, would be able to solve 3-\sat. 

Now suppose there is an $\exp(2^{\log^{\delta-\varepsilon}N})$-time $\log^{2-\delta-\varepsilon}k$-approximation algorithm for \gst,
for some $0<\varepsilon<\delta$.
We apply such algorithm to solve an instance of \gst derived from 3-\sat as above.
Then we have an algorithm that runs in time 
\begin{align*}
\exp(2^{\log^{\delta-\varepsilon}N})
= 
\exp(2^{(\log^{1/\delta+o(1)}{n})^{\delta-\varepsilon}})
= 
\exp(2^{(\log^{\frac{(1+o(1))(\delta-\varepsilon)}{\delta}}{n})})
=
\exp(2^{o(\log{n})})
= 2^{o({n})}\\
\end{align*}
This implies a subexponential-time algorithm for 3-\sat, which contradicts the $\ethh$.
Therefore, unless the $\ethh$ is false, there is no $\exp(2^{\log^{\delta-\varepsilon}N})$-time $\log^{2-\delta-\varepsilon}k$-approximation algorithm for \gst, thus proving Theorem~\ref{thm:gst-subexpo-hardness}

Since we take log from the expression, the above is also true 
if we replace $k$ by  $N$.
Combined with \cite{ChekuriP05} and \cite{EvenKS02},
we have the following corollary which shows 
almost tight running-time lower and upper bounds for approximating $\gst$ and $\cst$
to a factor $\log^{2-\delta}N$.
\begin{corollary}
The $\gst$ and {\cst} problems on graphs with $n$ vertices
admit $\log^{2-\delta}n$-approximation algorithms
for any constant $\delta<1$ that runs in time 
$\exp(2^{(1+o(1))\log^{\delta}n})$.
In addition, for any constant $\epsilon>1$
there is no $\log^{2-\delta-\epsilon} n$ approximation algorithm
for $\gst$ and $\cst$ that runs in time 
$\exp(2^{(1+o(1))\log^{\delta-\epsilon}n})$.
\end{corollary}

{\bf Remark:} We omit the algorithm for the {\cp}
problem, as its similar to the algorithm for $\dst$.
The lower bound holds because {\cp} 
has $\setcover$ as a special case.

\subparagraph*{Acknowledgments.}

The work of Marek Cygan is part of a project TOTAL that has received funding from the European Research Council (ERC)
under the European Union’s Horizon 2020 research and innovation programme (grant agreement No 677651).
Guy Kortsarz was partially supported by NSF grant 1540547.
Bundit Laekhanukit was partially supported by ISF grant no. 621/12, and I-CORE grant no. 4/11 and by the 1000-youth award from the Chinese Government.
Parts of the work were done while Guy Kortsarz and Bundit Laekhanukit were at the Weizmann Institute of Science, and parts of it were done while Bundit Laekhanukit was at the Max-Plank-Institute for Informatics.

\bibliography{u1}

\begin{thebibliography}{10}

\bibitem{BansalCLNN17-arxiv}
N.~Bansal, P.~Chalermsook, B.~Laekhanukit, D.~Nanongkai, and J.~Nederlof.
\newblock New tools and connections for exponential-time approximation.
\newblock {\em CoRR}, abs/1708.03515, 2017.

\bibitem{Bartal96}
Y.~Bartal.
\newblock Probabilistic approximations of metric spaces and its algorithmic
  applications.
\newblock In {\em 37th Annual Symposium on Foundations of Computer Science,
  {FOCS} '96, Burlington, Vermont, USA, 14-16 October, 1996}, pages 184--193,
  1996.

\bibitem{BonnetE0P15}
E.~Bonnet, B.~Escoffier, E.~J. Kim, and V.~T. Paschos.
\newblock On subexponential and fpt-time inapproximability.
\newblock {\em Algorithmica}, 71(3):541--565, 2015.
\newblock Preliminary version in IPEC'13.

\bibitem{BonnetLP18}
{\'{E}}.~Bonnet, M.~Lampis, and V.~T. Paschos.
\newblock Time-approximation trade-offs for inapproximable problems.
\newblock {\em J. Comput. Syst. Sci.}, 92:171--180, 2018.

\bibitem{BonnetP18}
{\'{E}}.~Bonnet and V.~T. Paschos.
\newblock Sparsification and subexponential approximation.
\newblock {\em Acta Inf.}, 55(1):1--15, 2018.

\bibitem{BourgeoisEP09a}
N.~Bourgeois, B.~Escoffier, and V.~T. Paschos.
\newblock Efficient approximation of min set cover by moderately exponential
  algorithms.
\newblock {\em Theor. Comput. Sci.}, 410(21-23):2184--2195, 2009.

\bibitem{BourgeoisEP11}
N.~Bourgeois, B.~Escoffier, and V.~T. Paschos.
\newblock Approximation of max independent set, min vertex cover and related
  problems by moderately exponential algorithms.
\newblock {\em Discrete Applied Mathematics}, 159(17):1954--1970, 2011.

\bibitem{CalinescuZ05}
G.~C{\u{a}}linescu and A.~Zelikovsky.
\newblock The polymatroid steiner problems.
\newblock {\em J. Comb. Optim.}, 9(3):281--294, 2005.
\newblock Preliminary version in ISAAC'04.

\bibitem{ChalermsookLN13}
P.~Chalermsook, B.~Laekhanukit, and D.~Nanongkai.
\newblock Independent set, induced matching, and pricing: Connections and tight
  (subexponential time) approximation hardnesses.
\newblock In {\em 54th Annual {IEEE} Symposium on Foundations of Computer
  Science, {FOCS} 2013, 26-29 October, 2013, Berkeley, CA, {USA}}, pages
  370--379, 2013.

\bibitem{CharikarCCDGGL99}
M.~Charikar, C.~Chekuri, T.~Cheung, Z.~Dai, A.~Goel, S.~Guha, and M.~Li.
\newblock Approximation algorithms for directed steiner problems.
\newblock {\em J. Algorithms}, 33(1):73--91, 1999.
\newblock Preliminary version in SODA'98.

\bibitem{ChekuriEK06}
C.~Chekuri, G.~Even, and G.~Kortsarz.
\newblock A greedy approximation algorithm for the group steiner problem.
\newblock {\em Discrete Applied Mathematics}, 154(1):15--34, 2006.

\bibitem{ChekuriP05}
C.~Chekuri and M.~P{\'{a}}l.
\newblock A recursive greedy algorithm for walks in directed graphs.
\newblock In {\em 46th Annual {IEEE} Symposium on Foundations of Computer
  Science {(FOCS} 2005), 23-25 October 2005, Pittsburgh, PA, USA, Proceedings},
  pages 245--253, 2005.

\bibitem{Chvatal79}
V.~Chv{\'{a}}tal.
\newblock A greedy heuristic for the set-covering problem.
\newblock {\em Math. Oper. Res.}, 4(3):233--235, 1979.

\bibitem{CyganKPW08}
M.~Cygan, L.~Kowalik, M.~Pilipczuk, and M.~Wykurz.
\newblock Exponential-time approximation of hard problems.
\newblock {\em CoRR}, abs/0810.4934, 2008.

\bibitem{CyganKW09}
M.~Cygan, L.~Kowalik, and M.~Wykurz.
\newblock Exponential-time approximation of weighted set cover.
\newblock {\em Inf. Process. Lett.}, 109(16):957--961, 2009.

\bibitem{Dinur07}
I.~Dinur.
\newblock The {PCP} theorem by gap amplification.
\newblock {\em J. {ACM}}, 54(3):12, 2007.
\newblock Preliminary version in STOC'06.

\bibitem{DinurS14}
I.~Dinur and D.~Steurer.
\newblock Analytical approach to parallel repetition.
\newblock In {\em Symposium on Theory of Computing, {STOC} 2014, New York, NY,
  USA, May 31 - June 03, 2014}, pages 624--633, 2014.

\bibitem{DreyfusW71}
S.~E. Dreyfus and R.~A. Wagner.
\newblock The {S}teiner problem in graphs.
\newblock {\em Networks}, 1(3):195--207, 1971.

\bibitem{EvenKS02}
G.~Even, G.~Kortsarz, and W.~Slany.
\newblock On network design problems: Fixed cost flows and the covering steiner
  problem.
\newblock In {\em Algorithm Theory - {SWAT} 2002, 8th Scandinavian Workshop on
  Algorithm Theory, Turku, Finland, July 3-5, 2002 Proceedings}, pages
  318--327, 2002.

\bibitem{FakcharoenpholRT04}
J.~Fakcharoenphol, S.~Rao, and K.~Talwar.
\newblock A tight bound on approximating arbitrary metrics by tree metrics.
\newblock {\em J. Comput. Syst. Sci.}, 69(3):485--497, 2004.
\newblock Preliminary version in STOC'03.

\bibitem{Feige98}
U.~Feige.
\newblock A threshold of ln \emph{n} for approximating set cover.
\newblock {\em J. {ACM}}, 45(4):634--652, 1998.
\newblock Preliminary in STOC'96.

\bibitem{FotakisLP16}
D.~Fotakis, M.~Lampis, and V.~T. Paschos.
\newblock Sub-exponential approximation schemes for csps: From dense to almost
  sparse.
\newblock In {\em 33rd Symposium on Theoretical Aspects of Computer Science,
  {STACS} 2016, February 17-20, 2016, Orl{\'{e}}ans, France}, pages
  37:1--37:14, 2016.

\bibitem{GargKR98}
N.~Garg, G.~Konjevod, and R.~Ravi.
\newblock A polylogarithmic approximation algorithm for the group steiner tree
  problem.
\newblock In {\em Proceedings of the Ninth Annual {ACM-SIAM} Symposium on
  Discrete Algorithms, 25-27 January 1998, San Francisco, California.}, pages
  253--259, 1998.

\bibitem{GuptaS03}
A.~Gupta and A.~Srinivasan.
\newblock On the covering steiner problem.
\newblock In {\em {FST} {TCS} 2003: Foundations of Software Technology and
  Theoretical Computer Science, 23rd Conference, Mumbai, India, December 15-17,
  2003, Proceedings}, pages 244--251, 2003.

\bibitem{HalperinK03}
E.~Halperin and R.~Krauthgamer.
\newblock Polylogarithmic inapproximability.
\newblock In {\em STOC}, pages 585--594, 2003.

\bibitem{Hastad96}
J.~H{\aa}stad.
\newblock Clique is hard to approximate within n\({}^{\mbox{1-epsilon}}\).
\newblock In {\em 37th Annual Symposium on Foundations of Computer Science,
  {FOCS} '96, Burlington, Vermont, USA, 14-16 October, 1996}, pages 627--636,
  1996.

\bibitem{ImpagliazzoP01}
R.~Impagliazzo and R.~Paturi.
\newblock On the complexity of k-sat.
\newblock {\em J. Comput. Syst. Sci.}, 62(2):367--375, 2001.
\newblock Preliminary version in CCC'99.

\bibitem{ImpagliazzoPZ01}
R.~Impagliazzo, R.~Paturi, and F.~Zane.
\newblock Which problems have strongly exponential complexity?
\newblock {\em J. Comput. Syst. Sci.}, 63(4):512--530, 2001.
\newblock Preliminary version in FOCS'98.

\bibitem{Johnson74a}
D.~S. Johnson.
\newblock Approximation algorithms for combinatorial problems.
\newblock {\em J. Comput. Syst. Sci.}, 9(3):256--278, 1974.
\newblock Preliminary version in STOC'73.

\bibitem{KortsarzP99}
G.~Kortsarz and D.~Peleg.
\newblock Approximating the weight of shallow steiner trees.
\newblock {\em Discrete Applied Mathematics}, 93(2-3):265--285, 1999.
\newblock Preliminary version in SODA'97.

\bibitem{LundY94}
C.~Lund and M.~Yannakakis.
\newblock On the hardness of approximating minimization problems.
\newblock {\em J. {ACM}}, 41(5):960--981, 1994.

\bibitem{Moshkovitz15}
D.~Moshkovitz.
\newblock The projection games conjecture and the np-hardness of ln
  n-approximating set-cover.
\newblock {\em Theory of Computing}, 11:221--235, 2015.
\newblock Preliminary version in APPROX'12.

\bibitem{MoshkovitzR10}
D.~Moshkovitz and R.~Raz.
\newblock Two-query {PCP} with subconstant error.
\newblock {\em J. {ACM}}, 57(5):29:1--29:29, 2010.
\newblock Preliminary version in FOCS'08.

\bibitem{Wolsey82}
L.~A. Wolsey.
\newblock An analysis of the greedy algorithm for the submodular set covering
  problem.
\newblock {\em Combinatorica}, 2(4):385--393, 1982.

\bibitem{Zuckerman07}
D.~Zuckerman.
\newblock Linear degree extractors and the inapproximability of max clique and
  chromatic number.
\newblock {\em Theory of Computing}, 3(1):103--128, 2007.
\newblock Preliminary version in STOC'06.

\end{thebibliography}

\end{document}